%% file: main.tex
\documentclass{IEEEtran}

\usepackage{cite}
\usepackage{amsmath,amssymb,amsfonts}
\usepackage{algorithmic}
\usepackage{graphicx}
\usepackage{textcomp}
\usepackage{xcolor}

\usepackage{amsthm}

\allowdisplaybreaks

\usepackage{tikz}
\usepackage{pgfplots}

\newtheorem{theorem}{Theorem}

\newtheorem{proposition}{Proposition}
\newtheorem{assumption}{Assumption}
\newtheorem{remark}{Remark}
\newtheorem{corollary}{Corollary}
\newtheorem{lemma}{Lemma}

\newtheorem{definition}{Definition}

\usetikzlibrary{shapes,arrows}
\usetikzlibrary{arrows,calc,positioning}
\tikzset{
	block/.style = {draw, rectangle,
		minimum height=1cm,
		minimum width=2cm},
	input/.style = {coordinate,node distance=1cm},
	output/.style = {coordinate,node distance=4cm},
	arrow/.style={draw, -latex,node distance=2cm},
	pinstyle/.style = {pin edge={latex-, black,node distance=2cm}},
	sum/.style = {draw, circle, node distance=1cm},
}

\pgfplotsset{compat=newest}
\usetikzlibrary{plotmarks}
\usetikzlibrary{arrows.meta}
\usepgfplotslibrary{patchplots}

\newcommand{\Z}{\ensuremath{{\mathbb Z}}}
\newcommand{\N}{\ensuremath{{\mathbb N}}}
\newcommand{\R}{\ensuremath{{\mathbb R}}}
\newcommand{\C}{\ensuremath{{\mathbb C}}}

\newcommand{\calB}{\ensuremath{\mathcal{B}}}
\newcommand{\calO}{\ensuremath{\mathcal{O}}}
\newcommand{\calI}{\ensuremath{\mathcal{I}}}
\newcommand{\calT}{\ensuremath{\mathcal{T}}}

\newcommand{\bfA}{\ensuremath{\mathbf{A}}}
\newcommand{\bfB}{\ensuremath{\mathbf{B}}}
\newcommand{\bfC}{\ensuremath{\mathbf{C}}}
\newcommand{\bfD}{\ensuremath{\mathbf{D}}}

\newcommand{\pp}{\ensuremath{\mathrm{p}}}
\newcommand{\ff}{\ensuremath{\mathrm{f}}}

\newcommand{\rmT}{\ensuremath{\mathrm{T}}}
\newcommand{\rmD}{\ensuremath{\mathrm{D}}}
\newcommand{\rmH}{\ensuremath{\mathrm{H}}}
\newcommand{\rmG}{\ensuremath{\mathrm{G}}}

\newcommand{\sfD}{\ensuremath{\mathsf{D}}}
\newcommand{\sfN}{\ensuremath{\mathsf{N}}}
\newcommand{\sfC}{\ensuremath{\mathsf{C}}}
\newcommand{\sfn}{\ensuremath{\mathsf{n}}}
\newcommand{\sfh}{\ensuremath{\mathsf{h}}}
\newcommand{\sfm}{\ensuremath{\mathsf{m}}}

\DeclareMathOperator{\im}{im}
\DeclareMathOperator{\rank}{rank}
\DeclareMathOperator*{\argmin}{argmin}

\DeclareMathOperator{\col}{col}
\DeclareMathOperator{\diag}{diag}

\def\BibTeX{{\rm B\kern-.05em{\sc i\kern-.025em b}\kern-.08em
    T\kern-.1667em\lower.7ex\hbox{E}\kern-.125emX}}

\begin{document}

\title{
Ensuring Stability of Non-Minimal Modes\\
in Input-Output Data-Driven Representation
}
\author{Joowon Lee, \IEEEmembership{Member, IEEE},
Nam Hoon Jo, \IEEEmembership{Member, IEEE},
Hyungbo Shim, \IEEEmembership{Fellow, IEEE},
Florian Dörfler,
\IEEEmembership{Fellow, IEEE},
and Jinsung Kim, \IEEEmembership{Member, IEEE}
\thanks{This work was supported by the grant from Hyundai Motor Company's R\&D Division and by the National Research Foundation of Korea(NRF) grant funded by the Korea government(MSIT) (No. RS-2022-00165417) \& (No. 2022R1F1A1074838).
    }
\thanks{J.~Lee is with Department of Decision and Control Systems, KTH Royal Institute of Technology, Stockholm 11428, Sweden (e-mail: joowon@kth.se).}
\thanks{N.\,H.~Jo is with Department of Electrical Engineering, Soongsil University, Seoul 06978, Korea (e-mail: nhjo@ssu.ac.kr).}
\thanks{H.~Shim is with ASRI, Department of Electrical and Computer Engineering, Seoul National University, Seoul 08826, Korea (e-mail: hshim@snu.ac.kr).}
\thanks{F.~Dörfler is with Department
of Information Technology and Electrical Engineering, ETH Zürich,
Zürich 8092, Switzerland (e-mail: dorfler@ethz.ch).}
\thanks{J.~Kim is with Advanced Vehicle Platform Division, Hyundai Motor Company, Seongnam 13529, Korea (e-mail: jinsung.kim@hyundai.com).}
}

\maketitle

\begin{abstract}
    Many recent data-driven control approaches for linear time-invariant systems are based on output trajectory prediction using input-output data matrices.
    The system dynamics described by this predictor, which we refer to as the input-output data-driven representation, yields non-unique autoregressive with exogenous inputs (ARX) models having possibly unstable non-minimal modes.
    In this note, we show that the stability of these non-minimal modes is ensured by a certain choice of ARX model, which coincides with the minimum-norm least-squares predictor using the Moore-Penrose inverse of the data matrix.
    This stability guarantee holds regardless of the underlying system's stability.
    Moreover, the stability persists under sufficiently small noise in data when a suitably truncated Moore-Penrose inverse is used.
    Consequently, the ARX model need not be reduced to the true system order in order to avoid unstable additional modes.
\end{abstract}

\begin{IEEEkeywords}
Data-driven control, system identification, linear systems.
\end{IEEEkeywords}

\section{Introduction}\label{sec:introduction}

Behavioral systems theory \cite{paradigms,TimeSeries} has served as a core theoretical basis for many recent data-driven control approaches,
particularly in methods for linear time-invariant (LTI) systems.
The principle of these methods lies in \emph{data-driven simulation} \cite{DDsimul},
which predicts the output of an unknown system from its input-output data identifying the input-output behavior;
this is apparent in, for example, data-enabled predictive control \cite{DeePC}.
Data-driven simulation solves a linear equation that consists of data matrices and the past and future input-output trajectories.
This indicates the system's input-output relation, which many data-driven control methods using input-output data are based on \cite{DDformulas,outputOptimal,prior,terminal,MIMO,noisyIO,KalmanLQ}.
In this note, we refer to this input-output relation described by data-driven simulation
as the \emph{input-output data-driven representation}\footnote{Precisely, this is \eqref{eq:DD_rep}, a valid system representation by Proposition~\ref{prop:prelim}. It can also be understood as a recursive application of data-driven simulation.}.

The input-output data-driven representation
induces a set of
linear difference equations with respect to inputs and outputs,
namely autoregressive with exogenous inputs (ARX) representations of the system.
These induced ARX representations are generally \emph{non-unique},
and in fact, may have different stability properties
although they are from the same data
identifying the system's input-output behavior
(for detail, see Section~\ref{sec:motiv}).
A numerical example demonstrates that
a \emph{stable} system's input-output data-driven representation can yield an \emph{unstable} ARX representation, as shown in Fig.~\ref{fig:forward}.

This phenomenon is due to
the presence of \emph{non-minimal modes} in the ARX representations, causing pole-zero cancellations that can either be stable or unstable.
They even appear as
uncontrollable modes in state-space representations
that are trivially derived from the ARX representations.
As this affects the follow-up controller synthesis,
existing data-driven control methods have either avoided the non-minimal modes through assumptions on the model knowledge or system class \cite{DDformulas,outputOptimal,prior,terminal,noisyIO},
or removed them through additional procedures \cite{MIMO,KalmanLQ}.
Both approaches align with the classical problem of
obtaining an ARX or state-space model
from input-output data,
without (unstable) non-minimal modes \cite{Ljung}.

This problem has been addressed in various ways,
including earlier methods as in \cite{Bell,minimal}, subspace identification methods \cite{subspaceBook},
the recent data-driven methods mentioned above, or \cite{shortest};
they either obtain a model of the true system order, which is estimated beforehand, or reduce the model order afterwards by eliminating the non-minimal modes.
Either way, the model is brought to the system order,
which is an additional operation on the input-output data-driven representation.
Yet, if a certain choice of its induced ARX representation ensures the stability of these non-minimal modes,
then
it is readily applicable to a wide range of data-driven control techniques;
for example, to
those based on recursive predictors, or on the \emph{extended} state-space representations constructed directly from the ARX form as in \cite{DDformulas,outputOptimal,prior,terminal,noisyIO,KalmanLQ}
(for detail, see Section~\ref{subsec:DD}).
Indeed, to the best of our knowledge,
it has not been previously shown that
the non-minimal modes are always stable when the ARX coefficients are obtained from the minimum-norm solution to an ordinary least-squares problem in terms of data.

We ensure the stability of the non-minimal modes by taking the Moore-Penrose inverse of the data matrix, or equivalently, by choosing the minimum-norm solution to the linear equation of data-driven simulation.
For multi-output systems, the ARX model is formulated by constructing
a separate minimum-norm least-squares predictor for each output,
without restricting the system class.
This stability guarantee holds regardless of
the number of the non-minimal modes, which is determined by a parameter choice for the representation, and even for unstable systems.
It is further shown, through a perturbation analysis, that this stability persists under sufficiently small noise in data,
when a suitably truncated Moore-Penrose inverse is used.
Overall, these results provide a theoretical basis for using non-minimal
ARX representations in data-driven modeling and control, instead of
reducing their order solely to ensure the stability of additional
modes.

The rest of this note is organized as follows.
Section~\ref{sec:pre} describes the problem setup and provides preliminaries
on data-driven simulation and behavioral systems theory.
Section~\ref{sec:motiv} discusses the motivation of this work, along with a numerical example.
Section~\ref{sec:main} presents our main results on the stability guarantee of the non-minimal modes.
Section~\ref{sec:noisy} considers the effect of noisy data.
Finally, Section~\ref{sec:conclusion} concludes the note.

\textit{Notations:} The sets of real numbers, complex numbers, positive integers, and nonnegative integers are denoted by $\R$, $\C$, $\N$, and $\Z_{\geq 0}$, respectively.
For $a_i\in\R^{m_i}$ with $i\in \lbrace 1,\,\ldots,\,n\rbrace$,
let
$\col( a_1,\,\ldots,\,a_n):=[a_1^\top,\,\ldots,\,a_n^\top]^\top$.
For $n\in\Z_{\geq 0}$,
$P_n$ denotes
the space of polynomials of degree no greater than $n$,
and
$M_n$
denotes
the set of monic polynomials of degree $n$.
We write the Moore-Penrose inverse and a generalized inverse of a matrix $A$ as $A^\dagger$ and $A^\rmG$, respectively.
Let $I_n\in\R^{n\times n}$,
$0_n\in\R^n$, and $0_{m\times n}\in\R^{m\times n}$
be the identity matrix, the zero vector, and the zero matrix,
respectively,
whose dimensions are omitted when clear from the context.
The diagonal matrix
with $a_i$ being the $i$-th main diagonal element is denoted by $\diag(a_1,\,\ldots,\,a_n)$;
it is also used for
rectangular diagonal matrices
whose dimensions are clear.
Let $\otimes$ denote the Kronecker product.
The smallest \emph{nonzero} singular value of a matrix $A$ is denoted by $\sigma_{\min}(A)$.

\section{Problem Setup and Preliminaries}
\label{sec:pre}

Consider a discrete-time LTI system described by the following state-space representation:
\begin{equation}\label{eq:sys_ss}
    \begin{aligned}
        x(t+1)&=Ax(t)+Bu(t),\\
        y(t)&=Cx(t)+Du(t),
    \end{aligned}
\end{equation}
where $x(t)\in\R^n$, $u(t)\in\R^m$, and $y(t)\in\R^p$ are the state, the input, and the output, respectively.
We denote the input-output behavior, also known as the \emph{external behavior} \cite{paradigms}, of \eqref{eq:sys_ss}
by
$\calB\subseteq (\R^{m+p})^{\Z_{\geq 0}}$, so that
$\lbrace \col(u(t),y(t))\rbrace_{t\geq 0}\in \calB$.

The external behavior $\calB$ is determined solely by the observable subsystem of \eqref{eq:sys_ss}.
Thus, for simplicity,
we assume that the pair $(C,A)$ is observable.
Otherwise, the results in this paper hold with respect to the observable subsystem of \eqref{eq:sys_ss}.
We note that the pair $(A,B)$ is not necessarily controllable.

\begin{remark}\label{rem:minimal}\upshape
    In terms of the external behavior $\calB$,
    the state-space representation \eqref{eq:sys_ss} is \emph{minimal}, i.e., has the smallest dimension among all state-space representations
    with behavior $\calB$, which is equivalent to the observability of $(C,A)$ \cite{TimeSeries,paradigms}.
    Unlike in minimal state-space realizations of transfer function matrices, this does not imply the controllability of $(A,B)$.\qed
\end{remark}

The \emph{lag} of the behavior $\calB$ is denoted by $l$, which coincides with the observability index of \eqref{eq:sys_ss} \cite{TimeSeries,paradigms}.
Note that $n\geq l$.

The system is unknown, which means that $(A,B,C,D)$, $n$, and $l$ are all unknown.

For $L\in\N$, define $\calB_L:=\calB\vert_{[0,L-1]}\subseteq \R^{(m+p)L}$,
i.e., the set of any possible input-output trajectories of length $L$ generated by \eqref{eq:sys_ss}.
For example,
$\col(u(0:L-1),\,y(0:L-1))\in \calB_L.$

Suppose that we have
input-output data matrices as
\begin{equation*}
    \rmD:=\begin{bmatrix}
        U_\pp\\ U_\ff \\ Y_\pp\\ Y_\ff
    \end{bmatrix}\in\R^{(m+p)(N+1)\times \rmT},\quad \rmH:=\begin{bmatrix}
        U_\pp\\ U_\ff \\ Y_\pp
    \end{bmatrix}
\end{equation*}
where
each column of $\rmD$ belongs to $\calB_{N+1}$ with parameter $N\in\N$, and
the matrices $U_\pp$, $U_\ff$, $Y_\pp$, and $Y_\ff$ have $mN$, $m$, $pN$, and $p$ rows, respectively.
These data are assumed to be noise-free,
whereas noisy data will be considered in Section~\ref{sec:noisy}.

\begin{lemma}\label{lem:H}
    For $N\geq l$, $\ker(\rmD)=\ker(\rmH)$.
\end{lemma}

\begin{proof}
    The proof is straightforward from the definition of lag (see, for example, \cite{fundamental}).
\end{proof}

\begin{proposition}\label{prop:prelim}
    Suppose that $N\geq l$.
    Then, $\im(\rmD)=\calB_{N+1}$ if and only if
    \begin{equation}\label{eq:rankH}
        \rank(\rmH)
    =m(N+1)+n.
    \end{equation}
    Furthermore, under the rank condition \eqref{eq:rankH}, given $u(t-N:t)$ and $y(t-N:t-1)$, $y(t)$ is uniquely determined by
    \begin{subequations}\label{eq:DD_rep}
        \begin{equation}\label{eq:DD_rep_y}
            y(t)=Y_\ff g
        \end{equation}
        for any $g\in\R^{\rmT}$ satisfying
        \begin{equation}\label{eq:DD_rep_Hg}
            \rmH g=\begin{bmatrix}
            u(t-N:t-1)\\ u(t)\\ y(t-N:t-1)
            \end{bmatrix}.
        \end{equation}
    \end{subequations}
\end{proposition}

\begin{proof}
    Since $n$ is the minimal state-space dimension in terms of $\calB$ (see Remark~\ref{rem:minimal}), also known as the \emph{state cardinality},
    the equivalence holds by
    \cite[Corollary~21]{identifiability}.
    The rest of the proof directly follows from Lemma~\ref{lem:H} and is thus omitted.
\end{proof}

Proposition~\ref{prop:prelim} provides a necessary and sufficient condition for the data matrices to identify the system's external behavior.
It also generalizes the \emph{data-driven simulation} proposed in \cite{DDsimul}; the system's future output trajectory can be perfectly simulated from the data matrices, given the past input-output trajectories and the future input trajectory.
Thus, the system's input-output behavior can be fully described by \eqref{eq:DD_rep}, and therefore we refer to \eqref{eq:DD_rep} as the \emph{input-output data-driven representation}.

\begin{assumption}\label{assum}
    Let $N\geq n$ and the rank condition \eqref{eq:rankH} hold.
\end{assumption}

In other words, we assume that the external behavior $\calB$ can be identified by the data matrices.
As discussed in Section~\ref{sec:motiv}, this does not necessarily mean that the system's state-space or ARX representation can be identified.

\begin{remark}\label{rem:ctrb}\upshape
    Given that the \emph{behavior $\calB$ is controllable}, the rank condition \eqref{eq:rankH} is satisfied if
    the data matrices have Hankel structure
    where the input data is persistently exciting of order $n+N+1$ (Willems' fundamental lemma \cite{fundamental}).
    Note that $\calB$ is controllable if and only if the state-space representation \eqref{eq:sys_ss} has no mode that is both uncontrollable and observable,
    which is in general weaker than the controllability of $(A,B)$. \qed
\end{remark}

\section{Motivation}
\label{sec:motiv}

We provide some key observations on the input-output data-driven representation \eqref{eq:DD_rep}, which gives motivation to analyze its stability.
Specifically, we demonstrate how this representation branches into non-unique ARX representations of the system
that may have different stability properties,
through a numerical example.

To begin with, observe that with $N\geq n$, there exists $h\in\R^{(mN+m+pN)\times p}$
such that
\begin{equation}\label{eq:h_eq}
    Y_\ff=h^\top \rmH,
\end{equation}
by Lemma~\ref{lem:H}.
Due to the rank condition \eqref{eq:rankH}, such $h$ is non-unique when $pN>n$.

For any $h$ satisfying \eqref{eq:h_eq}, every input-output trajectory of \eqref{eq:sys_ss} satisfies the following ARX representation:
\begin{equation}\label{eq:ARXh}
    y(t)=h^\top \begin{bmatrix}
            u(t-N:t-1)\\ u(t)\\ y(t-N:t-1)
        \end{bmatrix},
\end{equation}
by substituting $h^\top \rmH$ for $Y_\ff$ in \eqref{eq:DD_rep_y} and applying the equality \eqref{eq:DD_rep_Hg}.
By Proposition~\ref{prop:prelim}, recursive application of \eqref{eq:ARXh}, initialized with
an input-output history consistent with $\mathcal B$, reproduces
the system's output under the same future input.
When arbitrary initial histories are allowed, however, the ARX
relation may admit additional trajectories not belonging to
$\mathcal B$.
The stability of \eqref{eq:ARXh} may depend on the choice of $h$ among the solutions to \eqref{eq:h_eq},
as illustrated by the following numerical example.

\begin{remark}\label{rem:pN_n}\upshape
    By the rank condition \eqref{eq:rankH},
    a vector $h$ satisfying \eqref{eq:h_eq} is unique if and only if $pN=n$, which leads to a unique ARX representation \eqref{eq:ARXh}.
    This is related to why many existing data-driven control approaches assume $n=pl$ and let $N=l$ \cite{DDformulas,outputOptimal,prior,terminal,noisyIO}.
    This not only requires the exact knowledge of $n$ and $l$ but also restricts the class of multi-output systems.
    Our result obviates the need for such an assumption.
    \qed
\end{remark}

\subsection{Numerical Example}

\begin{figure}
	\centering
	\input{figures/forward}
	\caption{The output of the system \eqref{eq:MSD} (the black solid line) and $y(t)$
    obtained by applying \eqref{eq:ARXh} recursively for $h^\top=Y_\ff\rmH^\dagger$ (the blue dashed line) and $h^\top=Y_\ff\rmH^\rmG$ (the red dotted line).}
	\label{fig:forward}
\end{figure}

We considered
a single-input single-output (SISO), stable, controllable and observable mass-spring-damper system \cite{MSD};
\begin{align}
    \mathrm{m}_1\ddot{p}_1&=-\left(\mathrm{k}_0+\mathrm{k}_1\right)p_1 -\left(\mathrm{d}_0+\mathrm{d}_1\right)\dot{p}_1+\mathrm{k}_1p_2 \notag\\ &\qquad +\mathrm{d}_1\dot{p}_2+u, \notag \\
    \mathrm{m}_2\ddot{p}_2&=\mathrm{k}_1p_1+\mathrm{d}_1\dot{p}_1-\left(\mathrm{k}_1+\mathrm{k}_2\right)p_2-\left(\mathrm{d}_1+\mathrm{d}_2\right)\dot{p}_2, \notag\\
    y&=p_1,\label{eq:MSD}
\end{align}
where
$(\mathrm{m}_1,\mathrm{m}_2)=(10,9)$, $(\mathrm{k}_0,\mathrm{k}_1,\mathrm{k}_2)=(0.5,9,0.1)$, and $(\mathrm{d}_0,\mathrm{d}_1,\mathrm{d}_2)=(0.2,1.8,0.3)$.
The system \eqref{eq:MSD} was discretized under the sampling period $0.05$ s,
and has dimension
$n=4$.
The input-output data matrices were constructed with $N=6$ and $\rmT=94$.
Under the same data, we considered two cases,
when $h^\top=Y_\ff\rmH^\dagger$ and $h^\top=Y_\ff\rmH^\rmG$,
where $\rmH^\rmG$ is a randomly generated\footnote{\samepage When a matrix $A$ is written in the singular value decomposition, its (non-unique) generalized inverse can be written as $$A = U \begin{bmatrix} \Sigma & 0 \\ 0 & 0 \end{bmatrix} V^\top, \qquad A^\rmG = V \begin{bmatrix} \Sigma^{-1} & X \\ Y & Z \end{bmatrix} U^\top$$
where $U$ and $V$ are orthogonal, $\Sigma$ is diagonal with nonzero singular values, and $X$, $Y$, and $Z$ are any matrices of suitable sizes. We have randomly generated these $X$, $Y$, and $Z$ to obtain a generalized inverse.}
generalized inverse of $\rmH$.
Then, the output $y(t)$ was obtained by applying \eqref{eq:ARXh} recursively for each case,
as shown in Fig.~\ref{fig:forward}; the initial state was $(p_1,p_2,\dot{p}_1,\dot{p}_2)=(0,0.1,0,0)$ and the input was given as $u(t)=2\sin(0.05t)$.
The system's true outputs for $t\in[0,N)$ were provided to initially compute $y(N)$ from \eqref{eq:ARXh}.

By Proposition~\ref{prop:prelim},
recursive application of \eqref{eq:ARXh} yields correct prediction of the output, and thus
both $y(t)$'s in Fig.~\ref{fig:forward} should coincide with the system's true output.
However, the red dotted line in Fig.~\ref{fig:forward} deviates from the true output
and clearly shows a behavior of an unstable system.
This is caused by pole-zero cancellations occurring in \eqref{eq:ARXh};
in this example, they are stable
when $h^\top=Y_\ff\rmH^\dagger$ but unstable when $h^\top=Y_\ff\rmH^\rmG$.
In practice, the excessive poles and zeros cannot be perfectly canceled and
amplify numerical errors when they are unstable.

\section{Main Results}
\label{sec:main}

We first reveal the presence of non-minimal modes in the data-driven ARX representation \eqref{eq:ARXh},
whose stability may vary by the choice of $h$ satisfying \eqref{eq:h_eq}.
Next, we show that a certain choice of $h$, which is simply obtained via the Moore-Penrose inverse, ensures the stability of every non-minimal mode.
It is further discussed how this result leads to theoretical guarantees in data-driven modeling and control.
We begin our analysis in the case of single-output systems and then extend it to multi-output systems.

\subsection{Single-Output Case}\label{subsec:so}

Let $p=1$.
To begin with, we define
\begin{equation*}
    \frac{\sfN(z)}{\sfD(z)}:=C\left(zI-A\right)^{-1}B+D,
\end{equation*}
where $\sfD(z):=\det(zI-A)\in M_n$ so that $\sfN(z)\in P_n^{1\times m}$.
Any $\lambda\in\C$ such that $\sfD(\lambda)=0$ and $\sfN(\lambda)=0_m^\top$ is an uncontrollable mode of \eqref{eq:sys_ss}.
We also write
\begin{equation}\label{eq:h}
    h^\top=:\begin{bmatrix}
        \beta_0^\top & \cdots & \beta_{N-1}^\top & \beta_N^\top & \alpha_0 & \cdots & \alpha_{N-1}
    \end{bmatrix},
\end{equation}
where $\alpha_k\in\R$ for $k\in\lbrace 0,\,1,\,\ldots,\,N-1\rbrace$ and $\beta_k\in\R^m$ for $k\in\lbrace 0,\,1,\,\ldots,\,N\rbrace$.
Then, consider a mapping $\Phi:\lbrace h:Y_\ff=h^\top \rmH\rbrace\to M_{N-n}$ defined by
\begin{equation*}
    \Phi(h):=\sfC(z),
\end{equation*}
where $\sfC(z)$ satisfies
\begin{equation}\label{eq:C}
        z^N-\sum_{k=0}^{N-1}\alpha_k z^k=\sfC(z)\sfD(z),\quad
        \sum_{k=0}^N \beta_k^\top z^k=\sfC(z)\sfN(z).
\end{equation}

\begin{lemma}\label{lem:Phi}
    Under Assumption~\ref{assum}, the mapping $\Phi$ is well-defined and bijective.
\end{lemma}

\begin{proof}
    We first show that for any vector $h$ written by \eqref{eq:h} and satisfying \eqref{eq:h_eq}, there exists $\sfC(z)\in M_{N-n}$ such that \eqref{eq:C} holds.
    Observe that
    \begin{equation}\label{eq:lem_pf1}
        y(t-N:t-1)=\calO_N x(t-N)+\calI_N u(t-N:t-1),
    \end{equation}
    where $\calO_1:=C$, $\calI_1:=D$, and
    \begin{equation}\label{eq:calO}
    \calO_k:=\begin{bmatrix}
        \calO_{k-1} \\ CA^{k-1}
    \end{bmatrix}
    ,\quad \calI_k:=\begin{bmatrix}
            D & 0\\
            \calO_{k-1}B & \calI_{k-1}
        \end{bmatrix}
    \end{equation}
    for $k\geq 2$.
    By combining \eqref{eq:ARXh} and \eqref{eq:lem_pf1},
    we have
    \begin{align}\label{eq:lem_pf1_1}
        y(t)&=\left[\alpha_0,\,\ldots,\,\alpha_{N-1}\right]\!\left(\calO_Nx(t-N) + \calI_Nu(t-N:t-1)\right)\notag \\
        &\,\,+ \left[\beta_0^\top,\,\ldots,\,\beta_N^\top\right]u(t-N:t).
    \end{align}
    Meanwhile,
    \begin{align}\label{eq:lem_pf1_2}
        y(t)&=CA^Nx(t-N)+Du(t)\notag\\
        &\,\,+C\begin{bmatrix}
            A^{N-1}B & \cdots & AB & B
        \end{bmatrix} u(t-N:t-1).
    \end{align}
    Then, it follows that
    \begin{equation*}
        \left[\alpha_0,\,\ldots,\,\alpha_{N-1}\right]\calO_Nx_0=CA^N x_0
    \end{equation*}
    for any $x_0\in\R^{n}$,
    since the right hand sides of \eqref{eq:lem_pf1_1} and \eqref{eq:lem_pf1_2} are equal when the inputs are zero.
    Thus,
    \begin{equation}\label{eq:lem_pf3}
        C\alpha(A)=0_n^\top,\quad\text{where}\,\,\alpha(z):=z^N-\sum_{k=0}^{N-1}\alpha_k z^k.
    \end{equation}
    Since $N\geq n$, there exist polynomials $q(z)$ and $r(z)$, where the latter has degree less than $n$, such that
    $
        \alpha(z)=\sfD(z)q(z)+r(z)
    $.
    Then, by the Cayley-Hamilton theorem,
    $
        \alpha(A)=r(A).
    $
    Thus,
    $
        Cr(A)=0
    $ by \eqref{eq:lem_pf3}, which implies that
    $r(z)=0$ since $\calO_{n}$ is invertible.
    As a result, $\alpha(z)=\sfD(z)q(z)$ and $q(z)\in M_{N-n}$.
    Let $\sfC(z)=q(z)$.
    When $y(t)=0$ and $u(t)=0$ for all $t\leq 0$, i.e., when $x(0)=0$,
    applying the $z$-transform to \eqref{eq:ARXh} yields the transfer function matrix of \eqref{eq:sys_ss}, as
    \begin{equation*}
        \frac{1}{\alpha(z)}\sum_{k=0}^N \beta_k^\top z^k=\frac{1}{\sfC(z)\sfD(z)}\sum_{k=0}^N \beta_k^\top z^k=\frac{\sfN(z)}{\sfD(z)}.
    \end{equation*}
    Therefore, $\sfC(z)$ satisfies \eqref{eq:C} and $\Phi$ is well-defined. It is also injective by definition.

    Next, we show that $\Phi$ is surjective.
    If $N=n$, then by the rank condition \eqref{eq:rankH},
    both the domain of $\Phi$ and $M_0=\lbrace 1\rbrace$ are singletons, and thus
    $\Phi$ is surjective.
    Otherwise, it is clear that
    both the domain $\lbrace h:Y_\ff=h^\top \rmH\rbrace$ and $M_{N-n}$ are affine spaces associated with $\ker(\rmH^\top)$ and $P_{N-n-1}$, respectively.
    The mapping $\Phi$ is also affine;
    the mapping $\Phi^\prime:\ker(\rmH^\top)\to P_{N-n-1}$,
    defined by
    $
        \Phi^\prime(v):=\Phi(h+v)-\Phi(h)
    $
    for any $h$ satisfying \eqref{eq:h_eq}, is linear.
    Note that $\Phi^\prime$ is injective since $\Phi$ is injective.
    By the rank condition \eqref{eq:rankH}, both $\ker(\rmH^\top)$ and $P_{N-n-1}$ have dimension $N-n$,
    which implies that
    $\Phi^\prime$ is surjective.
    Consequently, $\Phi$ is surjective.
    This concludes the proof.
\end{proof}

We refer to the roots of the polynomial $\Phi(h)=\sfC(z)$ as the \emph{non-minimal modes} of the ARX representation \eqref{eq:ARXh}.
Therefore,
$\Phi$ can be understood as the mapping between \eqref{eq:ARXh} and its non-minimal modes.
As $\Phi$ is surjective by Lemma~\ref{lem:Phi}, it is always possible to obtain an ARX representation
from the data-driven representation \eqref{eq:DD_rep} that has non-minimal modes outside the unit circle when $N>n$.
In case of SISO systems, this means that any possible set of $N-n$ pole-zero cancellations can occur.

Note that not every common root of $z^N-\sum_{k=0}^{N-1}\alpha_k z^k$ and $
        \sum_{k=0}^N \beta_k^\top z^k$ in \eqref{eq:C} is a non-minimal mode,
as $\sfD(z)$ and $\sfN(z)$ may also have common roots.

Now we show that the non-minimal modes are always stable when $h^\top=Y_\ff\rmH^\dagger$, i.e.,
\begin{equation*}
    \sfC^\star(z):=\Phi((Y_\ff\rmH^\dagger)^\top)
\end{equation*}
is a Schur stable polynomial.
We utilize the fact that
\begin{equation}\label{eq:YH}
    (Y_{\ff}\rmH^\dagger)^\top=\argmin_{h\in\R^{m(N+1)+N}}\left\lVert h\right\rVert_2^2\qquad
    \text{subject to}\,\,\eqref{eq:h_eq}.
\end{equation}
As $\Phi$ is bijective, \eqref{eq:YH} can be rewritten as
\begin{equation}\label{eq:argminC}
    \sfC^\star(z)=\argmin_{\sfC(z)\in M_{N-n}}\left\lVert \Phi^{-1}(\sfC(z))\right\rVert_2^2.
\end{equation}
To utilize \eqref{eq:argminC},
we introduce a new cost function;
for $\sfn\in\N$ and $\tau\in\Z_{\geq 0}$,
$f_{\sfn,\tau}:M_\tau\times M_\sfn\times P_\sfn^{1\times m}\to \R$ is defined by
\begin{equation}\label{eq:f def}
    f_{\sfn,\tau}\left(r(z),p(z),q(z)\right)\!:=\!
    \left\lVert \!\left[\bar{q}_0^\top\!,\,\ldots,\,\bar{q}_{\sfn+\tau}^\top,\,\bar{p}_0,\,\ldots,\,\bar{p}_{\sfn+\tau-1}\right]^\top\!\right\rVert_2^2\!\!,
\end{equation}
where
\begin{equation*}
    r(z)p(z)=z^{\sfn+\tau}-\sum_{k=0}^{\sfn+\tau-1}\bar{p}_kz^k,\quad r(z)q(z)=\sum_{k=0}^{\sfn+\tau}\bar{q}_k^\top z^k.
\end{equation*}
With this definition in hand, we rewrite \eqref{eq:argminC} as
\begin{equation}\label{eq:Cstar}
    \sfC^\star(z)=\argmin_{\sfC(z)\in M_{N-n}} f_{n,N-n}\left(\sfC(z),\sfD(z),\sfN(z)\right).
\end{equation}

In what follows, we show that given polynomials
$p(z)\in M_\sfn$ and $q(z)\in P_\sfn^{1\times m}$,
any polynomial $r(z)\in M_\tau$ that minimizes the cost function \eqref{eq:f def} is Schur stable for all $\tau\in\N$;
then, the Schur stability of \eqref{eq:Cstar} is straightforward.
The proof considers a real root or a complex-conjugate pair
outside the open unit disk.
Factoring out such a root or pair reduces the required
comparison to optimization over a monic polynomial of
degree one or two.
We therefore first analyze these two cases.

\begin{lemma}\label{lem:tau12}
	Given $\sfn\in\N$, $p(z)\in M_\sfn$, and $q(z)\in P_\sfn^{1\times m}$, the followings hold:
    \begin{enumerate}
        \item
        Let $\displaystyle\lambda^\star:= \argmin_{\lambda\in\R}f_{\sfn,1}(z+\lambda,p(z),q(z))$. Then, $\lvert \lambda^\star\rvert<1$.
        \item
        Let $\displaystyle[\phi^\star,\,\psi^\star]^\top:=\argmin_{[\phi,\,\psi]^\top\in\R^2}f_{\sfn,2}(z^2+\phi z+\psi,p(z),q(z))$. Then, $\psi^\star <1$.
    \end{enumerate}
\end{lemma}

\begin{proof}
    See the Appendix.
\end{proof}

Based on Lemma~\ref{lem:tau12},
we derive the result for any $\tau\in\N$.

\begin{lemma}\label{lem:stable}
	Given $\sfn\in\N$, $p(z)\in M_\sfn$, and $q(z)\in P_\sfn^{1\times m}$,
    any $r(z)\in M_\tau$ that minimizes \eqref{eq:f def} is Schur stable for all $\tau\in\N$.
\end{lemma}

\begin{proof}
    Given $\tau\in\N$, let $r^\star(z)\in M_\tau$ minimize \eqref{eq:f def}.
    Suppose that there exists a root $\lambda\in\C$ of $r^\star(z)$ such that $\lvert \lambda\rvert \geq 1$.
    First, consider the case when $\lambda\in\R$.
    Then, $r^\star(z)$ can be factorized as $r^\star(z)=(z-\lambda)r_1(z)$, and by definition \eqref{eq:f def},
    \begin{multline}\label{eq:fpf1}
		f_{\sfn,\tau}\left(r^\star(z),p(z),q(z)\right)\\
		=f_{\sfn+\tau-1,1}(z-\lambda,r_1(z)p(z),r_1(z)q(z)).
    \end{multline}
    However, by Lemma~\ref{lem:tau12},
    the right hand side of \eqref{eq:fpf1} is strictly greater than when $\lambda$ is replaced with some $\lambda^\star$ such that $\lvert \lambda^\star\rvert<1$.
    Thus, \eqref{eq:f def} has less value when $r(z)=(z-\lambda^\star)r_1(z)$,
    which contradicts that $r^\star(z)$ minimizes \eqref{eq:f def}.
    Next, suppose $\lambda\notin \R$, which implies that $\tau >1$.
    Then, the complex conjugate $\bar{\lambda}$ of $\lambda$ is also a root of $r^\star(z)$,
    and
    therefore $r^\star(z)$ is factorized as
    $r^\star(z)=(z-\lambda)(z-\bar{\lambda})r_2(z)=:(z^2+\phi z+\psi)r_2(z)$, where $\psi=\lvert \lambda\rvert^2\geq 1$.
    The rest of the proof is analogous to that of the case when $\lambda\in\R$.
\end{proof}

\begin{remark}\upshape
    Employing $h^\top=Y_\ff\rmH^\dagger$ to reformulate the input-output data-driven representation \eqref{eq:DD_rep} into \eqref{eq:ARXh}
    shares significant overlap with applying the \emph{subspace predictor} \cite{SPC,ARC} to the data-driven simulation.
    It is also equivalent to assigning $g$ of \eqref{eq:DD_rep_y} as the least-norm solution to \eqref{eq:DD_rep_Hg}, which can be enforced by projection-based regularization schemes \cite{regularization}.\qed
\end{remark}

\subsection{General Case}\label{subsec:general}

Now we extend the analysis in Section~\ref{subsec:so} to multi-output systems.
To this end,
we consider the multi-output system \eqref{eq:sys_ss} as a set of parallel single-output systems, as
\begin{equation}\label{eq:sys_mo}
        x(t+1)=Ax(t)+Bu(t),\quad
        y_i(t)=C_ix(t)+D_iu(t),
\end{equation}
for $i\in \lbrace 1,\,\ldots,\,p\rbrace$,
where $y_i(t)$ is the $i$-th output and $C_i$ and $D_i$ are the $i$-th rows of $C$ and $D$, respectively.
Accordingly,
the data matrices are reconstructed for each $i$ as
\begin{equation*}
    Y_{\pp,i}:=\left(I_N\otimes e_i^\top\right)Y_\pp,\quad
    Y_{\ff,i}:=e_i^\top Y_\ff,\quad \rmH_i:=\begin{bmatrix}
        U_\pp\\ U_\ff \\ Y_{\pp,i}
    \end{bmatrix},
\end{equation*}
where $e_i\in\R^p$ is the unit vector whose $i$-th component is $1$.
We aim to show that the non-minimal modes in the following data-driven ARX representation are always stable:
\begin{subequations}\label{eq:DDR MP}
    \begin{align}
        y_i(t)&=Y_{\ff,i}\rmH_i^\dagger\begin{bmatrix}
        u(t-N:t-1)\\ u(t)\\ y_i(t-N:t-1)
        \end{bmatrix},\,\, i\in \lbrace 1,\,\ldots,\,p\rbrace,\label{eq:DDR MP i}\\
        y(t)&=\left[y_1(t),\,\ldots,\,y_p(t)\right]^\top.
    \end{align}
\end{subequations}
Here, the $i$-th output predictor uses all input channels and the past $i$-th outputs, but not the other output channels.

To begin with, we investigate if \eqref{eq:DDR MP} is a valid representation of the system under Assumption~\ref{assum}.
For each $i$, consider the observable subsystem of \eqref{eq:sys_mo}, with matrices $(A_{i,o},B_{i,o},C_{i,o},D_i)$ and dimension $n_i$. Note that $n_i$ coincides with the rank of the observability matrix of $(C_i,A)$.

\begin{lemma}\label{lem:noi}
    Under Assumption~\ref{assum}, for every $i\in\lbrace 1,\,\ldots,\,p\rbrace$, $N\geq n_i$ and $\rank(\rmH_i)=m(N+1)+n_i$.
\end{lemma}

\begin{proof}
    For a fixed $i\in\lbrace 1,\,\ldots,\,p\rbrace$,
    the observability matrix of $(C_i,A)$, which has rank $n_i$,
    is a result of
    eliminating certain rows from $\calO_n$ defined in \eqref{eq:calO}, which has rank $n$.
    Thus, $n_i\leq n\leq N$ by Assumption~\ref{assum}.
    To inspect the rank of $\rmH_i$, let $\calB^i$ be the external behavior of \eqref{eq:sys_mo}.
    Then, it is clear that if $\im(\rmD)=\calB_{N+1}$, then $\im(\rmD_i):=\im\left(\begin{bmatrix}
        \rmH_i\\ Y_{\ff,i}
    \end{bmatrix}\right)=\calB_{N+1}^i$.
    Since $n_i$ is also the observability index of the single-output system \eqref{eq:sys_mo},
    or the lag of $\calB^i$,
    Proposition~\ref{prop:prelim} is applicable, and thus
    $\rank(\rmH_i)=m(N+1)+n_i$.
\end{proof}

Lemma~\ref{lem:noi} together with Proposition~\ref{prop:prelim} implies that, when initialized with
an input-output trajectory pair consistent with $\mathcal B$,
\eqref{eq:DDR MP} reproduces the system's outputs under
any future input.

To state the main theorem, we define
\begin{equation*}
    \frac{\sfN_i(z)}{\sfD_i(z)}:=C_{i,o}\left(zI-A_{i,o}\right)^{-1}B_{i,o}+D_i,\quad i\in\lbrace 1,\,\ldots,\,p\rbrace,
\end{equation*}
where $\sfD_i(z):=\det(zI-A_{i,o})\in M_{n_i}$ so that $\sfN_i(z)\in P_{n_i}^{1\times m}$.
We also write each $Y_{\ff,i}\rmH_i^\dagger$ in \eqref{eq:DDR MP} as
\begin{equation}\label{eq:DDR MP coeff}
	Y_{\ff,i}\rmH_i^\dagger
	=:\left[b_{0,i}^\top,\,\ldots,\,b_{N-1,i}^\top,\,b_{N,i}^\top,\,a_{0,i},\,\ldots,\,a_{N-1,i}\right],
\end{equation}
where $a_{k,i}\in \R$ and $b_{k,i}\in \R^m$ for each $k$.

\begin{theorem}\label{thm:main}
    Suppose that Assumption~\ref{assum} holds, and consider the ARX representation \eqref{eq:DDR MP} whose coefficients are defined by \eqref{eq:DDR MP coeff}.
    For each $i\in \lbrace 1,\,\ldots,\,p\rbrace$, there exists a unique monic polynomial $\sfC_i^\star(z)$ of degree $N-n_i$ such that
    \begin{equation}\label{eq:thm i}
    z^N- \sum_{k=0}^{N-1}a_{k,i}z^k=\sfC_i^\star(z)\sfD_i(z),\quad
    \sum_{k=0}^{N}b_{k,i}^\top z^k=\sfC_i^\star(z)\sfN_i(z),
    \end{equation}
    and every root of $\sfC_i^\star(z)$ lies strictly inside the unit circle.
\end{theorem}

\begin{proof}
    Let $i\in \lbrace 1,\,\ldots,\,p\rbrace$ be fixed.
    The existence of $\sfC_i^\star(z)\in M_{N-n_i}$ satisfying \eqref{eq:thm i} follows from Lemmas~\ref{lem:Phi} and \ref{lem:noi}.
    Analogously to Section~\ref{subsec:so},
    it can be shown that
    \begin{equation*}
        \sfC_i^\star(z)=\argmin_{\sfC(z)\in M_{N-n_i}} f_{n_i,N-n_i}\left(\sfC(z),\sfD_i(z),\sfN_i(z)\right),
    \end{equation*}
    which is therefore Schur stable by Lemma~\ref{lem:stable}.
\end{proof}

We emphasize that Theorem~\ref{thm:main} ensures the stability of every non-minimal mode in \eqref{eq:DDR MP},
that is, every root of $\sfC_i^\star(z)$ for all $i\in \lbrace 1,\,\ldots,\,p\rbrace$, regardless of the system's stability.
Theorem~\ref{thm:main} holds for every $N\ge n$ under the rank condition \eqref{eq:rankH}.
Hence, choosing an ARX horizon larger than the true system
order does not introduce unstable non-minimal modes in the
noise-free representation \eqref{eq:DDR MP}.
In what follows, we further discuss the impact of this guarantee.

\subsection{Implications for Data-Driven Modeling and Control}\label{subsec:DD}

We have seen that data identifying the system's input-output behavior, or equivalently, satisfying the rank condition \eqref{eq:rankH}, do not necessarily identify its ARX representation.
Among these non-unique ARX representations from data, we propose \eqref{eq:DDR MP} as a preferable candidate to be used for data-driven modeling and analysis behind control,
owing to Theorem~\ref{thm:main}.

For example, as depicted in Fig.~\ref{fig:forward}, the system's stability is generally not sufficient for the stability of the ARX representation \eqref{eq:ARXh},
or the recursive application of data-driven simulation (output prediction) to operate properly.
Yet, by Theorem~\ref{thm:main}, it is sufficient when we employ \eqref{eq:DDR MP}.

\begin{corollary}\label{cor:ypred}
    Let $\hat{y}(t)$ be the predicted output by recursively applying \eqref{eq:DDR MP}.
    Then, under Assumption~\ref{assum}, $\lVert y(t)-\hat{y}(t)\rVert \to 0$ as $t\to \infty$ for any initial prediction $\hat{y}(-N:-1)\in\R^{pN}$ (which need not coincide with the measured outputs $y(-N:-1)$)
    if the system \eqref{eq:sys_ss} is stable.
\end{corollary}

\begin{proof}
    It can be easily seen that the prediction error dynamics is stable if and only if the polynomial $z^N- \sum_{k=0}^{N-1}a_{k,i}z^k$ is Schur stable for every $i\in \lbrace 1,\,\ldots,\,p\rbrace$.
    By Theorem~\ref{thm:main}, this is equivalent to all $\sfD_i(z)$'s, or $A_{i,o}$'s, being Schur stable.
    Since for each $i$, $A_{i,o}$ is a diagonal block of $A$ in the observability decomposition of $(C_i,A)$,
    every eigenvalue of $A_{i,o}$ is also an eigenvalue of $A$.
    Hence, $A_{i,o}$ is Schur stable if $A$ is Schur stable,
    which concludes the proof.
\end{proof}

\begin{remark}\label{rem:DDinv}\upshape
    A principle similar to the convergence of recursive data-driven simulation can be applied to data-driven unknown input reconstruction \cite{DDinv22,DDinv23,DDinv26},
    which recursively estimates the unknown input from output measurements using input-output data matrices.
    An analogous result to Corollary~\ref{cor:ypred}, with respect to the input estimate, needs the stability of every
    non-minimal mode in the data-driven representation of the inverse system.
    A recent result \cite{DDinv26} guarantees such stability through a projection within the data matrices using the Moore-Penrose inverse.\qed
\end{remark}

More importantly, the unknown stability of the non-minimal modes has complicated many data-driven control approaches,
particularly when they are based on input-output data.
This is because they regarded the \emph{extended state-space representation}
constructed from data,
where the state variable is
\begin{equation*}
    \chi(t):=\begin{bmatrix}
        y(t-N:t-1)\\ u(t-N:t-1)
    \end{bmatrix},
\end{equation*}
as a virtual plant;
even \emph{direct} approaches that bypass modeling in controller synthesis often rely on this extended representation for derivation and analysis \cite{DDformulas,outputOptimal,prior,terminal,MIMO,noisyIO,KalmanLQ}.
The problem is that such extended representations, derived from \eqref{eq:ARXh}, are generally
uncontrollable \cite{MIMO,noisyIO}
if there exist non-minimal modes.
In this regard, prior works have either avoided the non-minimal modes by assumptions (see Remark~\ref{rem:pN_n}) or introduced additional procedures to obtain controllable realizations \cite{MIMO,KalmanLQ}.

Specifically for single-output systems, we can construct the following state-space representation with $h$ written by \eqref{eq:h}:
\begin{equation}\label{eq:extended}
        \chi(t+1)=\bfA \chi(t)+\bfB u(t),\quad
        y(t)=\bfC \chi(t)+\bfD u(t),
\end{equation}
where
\begin{align*}
    \scalebox{0.9}{$\bfA$}&:=\!\scalebox{0.88}{$\begin{bmatrix} 0 & 1 & &  & & &  \\
			\vdots & & \ddots & & & & & \\
			0 & & & 1 & & & & \\
			\alpha_{0} & \alpha_{1} & \cdots & \alpha_{N-1} & \beta_{0}^\top & \beta_{1}^\top & \cdots & \beta_{N-1}^\top \\
			& & & & 0 & I_m & &  \\
			& & & & \vdots & & \ddots & \\
			& & & & 0 & & & I_m \\
			& & & & 0 & & \cdots & 0
        \end{bmatrix}$}\!,\,
    \scalebox{0.9}{$\bfB$}:=\!\scalebox{0.88}{$\begin{bmatrix}
        0 \\ \vdots \\ 0 \\ \beta_N^\top\\ 0\\ \vdots \\ 0\\ I_m
    \end{bmatrix}$}\!,\\
    \scalebox{0.9}{$\bfC$}&:=\!\scalebox{0.88}{$\begin{bmatrix}
        \alpha_{0} & \alpha_{1} & \cdots & \alpha_{N-1} & \beta_{0}^\top & \beta_{1}^\top & \cdots & \beta_{N-1}^\top
    \end{bmatrix}$},\,\scalebox{0.9}{$\bfD$}:=\!\scalebox{0.88}{$\beta_N^\top$}.
\end{align*}
Since \eqref{eq:extended} realizes the ARX relation \eqref{eq:ARXh}, it has the same transfer function matrix as \eqref{eq:sys_ss} by the proof of Lemma~\ref{lem:Phi} while its dimension is $N(m+1)$; hence it is a non-minimal realization of \eqref{eq:sys_ss}.

\begin{proposition}\label{prop:stabil}
    Let $p=1$.
    Then, under Assumption~\ref{assum}, the followings hold:
    \begin{enumerate}
        \item For any vector $h$ satisfying \eqref{eq:h_eq} and written by \eqref{eq:h},
        the uncontrollable modes of \eqref{eq:extended} consist of the roots of $\Phi(h)$ and the uncontrollable modes of \eqref{eq:sys_ss}.
        \item Let $h^\top=Y_\ff\rmH^\dagger$. Then,
        $(\bfA,\bfB)$ is stabilizable if and only if $(A,B)$ is stabilizable.
    \end{enumerate}
\end{proposition}

\begin{proof}
    It can be shown that given $\lambda\in\C$,  $v\bfA=\lambda v$ and $v\bfB=0_m^\top$ for some nonzero $v^\top\in\C^{(m+1)N}$ if and only if
    $\lambda^N=\sum_{k=0}^{N-1}\alpha_k \lambda^k$ and $\sum_{k=0}^N \beta_k^\top \lambda^k=0_m^\top$.
    Thus,
    1) directly follows from \eqref{eq:C} and Lemma~\ref{lem:Phi}, since the common roots of $\sfD(z)$ and $\sfN(z)$ are the uncontrollable modes of \eqref{eq:sys_ss}.
    Then, 2) holds by Theorem~\ref{thm:main}.
\end{proof}

Proposition~\ref{prop:stabil} implies that if the behavior $\calB$ is controllable (see Remark~\ref{rem:ctrb}), as many works on data-driven control assume,
then $(\bfA,\bfB)$ is \emph{always stabilizable} when $h^\top=Y_\ff\rmH^\dagger$.
Note that $(\bfC,\bfA)$ is detectable for any $h$ due to its structure.
Therefore,
data-driven controller design can be readily derived from the extended representation \eqref{eq:extended} using $h^\top=Y_\ff\rmH^\dagger$,
thanks to
our stability guarantee of the non-minimal modes.

\section{Effect of Noisy Data on Non-Minimal Modes}\label{sec:noisy}

We show that the stability established in Theorem~\ref{thm:main} persists under sufficiently small perturbations of the data matrices when a suitably truncated Moore-Penrose inverse is used.
For simplicity, we consider single-output systems, as the analysis applies separately to each $i$ in \eqref{eq:DDR MP}.

Suppose that we have noisy data matrices
$$\tilde\rmH := \rmH + E, \qquad \tilde Y_\ff := Y_\ff + e,$$
where $E$ and $e$ are perturbations induced by output measurement noises.
Since the roots of a monic polynomial depend continuously on its coefficients, it suffices to ensure that the ARX coefficients obtained from these noisy matrices converge to $Y_\ff\rmH^\dagger$ as the noise vanishes.
However, the direct use of $\tilde Y_\ff\tilde\rmH^\dagger$ does not guarantee this convergence; noise generally changes the rank of $\rmH$, under which the Moore-Penrose inverse loses continuity \cite{MPinv}.
We therefore use the truncated Moore-Penrose inverse, defined as follows.

\begin{definition}\label{def:truncMP}
    Let $A\in\R^{\sfm\times \sfn}$ have the singular value decomposition (SVD) $A=U\Sigma V^\top$,
    where
    $\Sigma:=\diag(\sigma_1,\,\ldots,\,\sigma_\sfh)$
    with $\sfh:=\min\lbrace \sfm,\,\sfn\rbrace$ and
    $\sigma_i\geq \sigma_{i+1}$ for $i\in\lbrace 1,\,\ldots,\,\sfh\rbrace$ ($\sigma_{\sfh+1}:=0$).
    Then,
    the truncated Moore-Penrose inverse
    of $A$ with tolerance $\tau>0$
    is defined by
    \begin{equation*}
        [A_\tau]^\dagger:=V \begin{bmatrix}
            \Sigma_\tau^{-1} & 0\\
            0 & 0
        \end{bmatrix}U^\top,\quad \text{where}\,\,\Sigma_\tau:=\diag(\sigma_1,\,\ldots,\,\sigma_\kappa)
    \end{equation*}
    with
    $\kappa$
    such that
    $\sigma_\kappa >\tau \geq \sigma_{\kappa+1}$.
\end{definition}

We consider \eqref{eq:DDR MP} with
\begin{equation}\label{eq:noisy YH}
    \tilde{Y}_\ff[\tilde{\rmH}_\tau]^\dagger=\left[\tilde{b}_0^\top,\,\ldots,\,\tilde{b}_{N-1}^\top,\,\tilde{b}_N^\top,\,\tilde{a}_0,\,\ldots,\,\tilde{a}_{N-1}\right],
\end{equation}
where $\tilde{b}_k\in\R^m$ and $\tilde{a}_k\in\R$ for each $k$.
Then,
the non-minimal modes with \eqref{eq:noisy YH} are regarded as the roots of $z^N-\sum_{k=0}^{N-1}\tilde{a}_k z^k$ that are close to those of $\sfC^\star(z)$, as specified below.

\begin{proposition}\label{prop:noisy}
    Under Assumption~\ref{assum}, let $p_1,\,\ldots,\,p_n$ be the roots of $\sfD(z)$, and let $q_1,\,\ldots,\,q_{N-n}$ be the roots of $\sfC^\star(z)$.
    If $0<\tau<\sigma_{\min}(\rmH)$, then for any $\epsilon>0$, there exists $\delta>0$ such that
    \begin{equation}\label{eq:noisy thm}
        z^N-\sum_{k=0}^{N-1}\tilde{a}_kz^k=\prod_{k=1}^{n}(z-\tilde{p}_k)\prod_{j=1}^{N-n}(z-\tilde{q}_j)
    \end{equation}
    with $\lvert\tilde{p}_k-p_k\rvert<\epsilon$ for each $k$
    and $\lvert\tilde{q}_j-q_j\rvert<\epsilon$ for each $j$
    if $\lVert\begin{bmatrix}E^\top & e^\top\end{bmatrix}\rVert_2<\delta$.
\end{proposition}

\begin{proof}
    Let $r:=\rank(\rmH)$,
    and let $\sigma_i$ and $\tilde{\sigma}_i$ be the $i$-th largest singular values of $\rmH$ and $\tilde{\rmH}$, respectively.
    It follows from the Weyl's theorem that $\lvert\tilde{\sigma}_i-\sigma_i\rvert\leq\lVert E\rVert_2$ for each $i$.
    Since $0<\tau<\sigma_{\min}(\rmH)=\sigma_r$, we have
    $\tilde{\sigma}_r > \tau > \tilde{\sigma}_{r+1}$ if
    $$\lVert E\rVert_2<\delta_0:=\min\lbrace\tau,\,\sigma_r-\tau\rbrace,$$
    and thus $\kappa=r$ for $[\tilde{\rmH}_\tau]^\dagger$ according to Definition~\ref{def:truncMP}.
    Let $\tilde{\rmH}_\tau$ be obtained from $\tilde{\rmH}$ by retaining only the singular values greater than $\tau$.
    Then, $[\tilde{\rmH}_\tau]^\dagger$ is the Moore-Penrose inverse of $\tilde{\rmH}_\tau$, and $\rank(\tilde{\rmH}_\tau)=\rank(\rmH)$.
    Since
    \begin{equation*}
        \lVert\tilde{\rmH}_\tau-\rmH\rVert_2\leq\lVert\tilde{\rmH}_\tau-\tilde{\rmH}\rVert_2+\lVert E\rVert_2=\tilde{\sigma}_{r+1}+\lVert E\rVert_2\leq 2\lVert E\rVert_2,
    \end{equation*}
    we have $[\tilde{\rmH}_\tau]^\dagger\to\rmH^\dagger$ as $\lVert E\rVert_2\to 0$,
    by the continuity of the Moore-Penrose inverse on matrices of fixed rank \cite{MPinv}.
    Thus,
    \begin{align*}
        \lVert\tilde{Y}_\ff[\tilde{\rmH}_\tau]^\dagger-Y_\ff\rmH^\dagger\rVert_2
        &\leq\lVert e\rVert_2\lVert[\tilde{\rmH}_\tau]^\dagger\rVert_2+\lVert Y_\ff\rVert_2\lVert[\tilde{\rmH}_\tau]^\dagger-\rmH^\dagger\rVert_2\\
        &<\lVert e\rVert_2/\tau+\lVert Y_\ff\rVert_2\lVert[\tilde{\rmH}_\tau]^\dagger-\rmH^\dagger\rVert_2,
    \end{align*}
    which tends to zero as $\lVert\begin{bmatrix}E^\top & e^\top\end{bmatrix}\rVert_2\to 0$.
    Consequently, the monic polynomial $z^N-\sum_{k=0}^{N-1}\tilde{a}_kz^k$ converges to $\sfD(z)\sfC^\star(z)$ coefficient-wise.
    The claim directly follows from the continuity of the roots of monic polynomials of fixed degree, with some $\delta\in(0,\delta_0]$.
\end{proof}

Theorem~\ref{thm:main} ensures that the roots $q_j$ lie strictly inside the unit circle.
Proposition~\ref{prop:noisy} therefore implies that the corresponding modes $\tilde{q}_j$'s in \eqref{eq:noisy thm} remain Schur stable under sufficiently small noise, regardless of the stability of the underlying system.
This conclusion holds for each $N\geq n$ satisfying Assumption~\ref{assum}; in particular, the ARX representation need not be reduced to the true system order to avoid any additional mode.

Proposition~\ref{prop:noisy} concerns the truncated Moore-Penrose inverse of $\tilde{\rmH}$ satisfying $\rank([\tilde{\rmH}_\tau]^\dagger)=\rank(\rmH^\dagger)$.
It does not address
whether the corresponding modes $\tilde{q}_j$'s remain stable when $\tilde{\rmH}$ is less truncated, that is, when $\rank([\tilde{\rmH}_\tau]^\dagger)>\rank(\rmH^\dagger)$.
Note that although the desired truncation restores the rank of $\rmH$, the ARX representation using \eqref{eq:noisy YH} retains the order $N$ and hence its $N-n$ corresponding modes.

\section{Conclusions}\label{sec:conclusion}

We have studied the stability of non-minimal modes in the input-output data-driven representation, which generally leads to non-unique ARX representations.
For noise-free data, the minimum-norm construction yields
 stable non-minimal modes, with the multi-output result
obtained through an output-wise construction.
A perturbation analysis further shows that the additional
modes remain stable under sufficiently small noise in data
when a suitably truncated Moore-Penrose inverse is used.
The significance of this result
has also been discussed in the context of system identification and data-driven control.

Yet, some aspects of this result remain to be explored.
For multi-output systems, a natural direction is to investigate the non-minimal modes of predictors that jointly use the past measurements of all outputs.
In addition, the effect of these non-minimal modes on data-driven control synthesis,
particularly in terms of control performance,
could be investigated.

\appendix[Proof of Lemma~\ref{lem:tau12}]

Let $p(z)=z^\sfn+\sum_{k=0}^{\sfn-1}p_kz^k$ and $q(z)=\sum_{k=0}^{\sfn}q^\top_kz^k$.
We define Toeplitz matrices for $p(z)$ and $q(z)$ by
\begin{equation*}
    \left(P_\calT,Q_\calT \right):=\left(\scalebox{0.95}{$\begin{bmatrix}
        1 & 0 & 0\\
        p_{\sfn-1} & 1 & 0\\
        p_{\sfn-2} & p_{\sfn-1} & 1\\
        p_{\sfn-3} & p_{\sfn-2} & p_{\sfn-1}\\
        \vdots & \vdots & \vdots\\
        p_0 & p_1& p_2\\
        0 & p_0 & p_1\\
        0 & 0 & p_0
    \end{bmatrix}$}, \scalebox{0.88}{$\begin{bmatrix}
        q_\sfn & 0_m & 0_m\\
        q_{\sfn-1} & q_\sfn & 0_m\\
        q_{\sfn-2} & q_{\sfn-1} & q_\sfn\\
        \vdots & \vdots & \vdots\\
        q_0 & q_1 & q_2 \\
        0_m & q_0 & q_1\\
        0_m & 0_m & q_0
    \end{bmatrix}$}\right)
\end{equation*}
so that $P_\calT\in\R^{(\sfn+3)\times 3}$
and $Q_\calT\in\R^{m(\sfn+3)\times 3}$.
We utilize the fact that multiplication of polynomials can be represented by a product between a Toeplitz matrix and a vector of coefficients.
Specifically,
for a polynomial $w(z)=w_2z^2+w_1z+w_0$,
\begin{equation}\label{eq:Toeplitz}
    P_\calT
    \col(w_2,\,w_1,\,w_0)=\col(\bar{w}_{\sfn+2},\,\bar{w}_{\sfn+1},\,\ldots,\,\bar{w}_0),
\end{equation}
where $\bar{w}_k$ is the $k$-th coefficient of $p(z)w(z)$.
We also define
\begin{equation*}
    \begin{bmatrix}
        P_\calT^\top & Q_\calT^\top
    \end{bmatrix}\begin{bmatrix}
        P_\calT\\ Q_\calT
    \end{bmatrix}=:\scalebox{0.95}{$\begin{bmatrix}
        \theta_0 & \theta_1& \theta_2\\
        \theta_1 & \theta_0 & \theta_1 \\
        \theta_2 & \theta_1 & \theta_0
    \end{bmatrix}$}.
\end{equation*}
We first prove 1).
From the property \eqref{eq:Toeplitz}, it follows that
\begin{align*}
    f_{\sfn,1}(z+\lambda,p(z),q(z))&=\left \lVert \begin{bmatrix}
        P_\calT\\ Q_\calT
    \end{bmatrix}
    \col(0,1,\lambda)
    \right\rVert_2^2-1\\
    &=\theta_0\lambda^2+2\theta_1\lambda+\theta_0-1,
\end{align*}
and thus $\lambda^\star=-\theta_1/\theta_0$.
By the Cauchy-Schwarz inequality,
\begin{equation*}
    \lvert \theta_1\rvert\leq \sqrt{\theta_0-1-q_\sfn^\top q_\sfn}\sqrt{\theta_0-p_0^2-q_0^\top q_0}<\theta_0,
\end{equation*}
since
\begin{multline*}
    \theta_1=\left[p_{\sfn-1},\,p_{\sfn-2},\,\ldots,\,p_0,\,q_{\sfn-1}^\top,\,\ldots,\,q_0^\top\right]\\
    \cdot\left[1,\,p_{\sfn-1},\,\ldots,\,p_1,\,q_\sfn^\top,\,\ldots,\,q_1^\top\right]^\top.
\end{multline*}
Therefore, $\lvert\lambda^\star\rvert <1$,
which concludes the proof of 1).

Next, we prove 2).
It can be observed that
\begin{align}
    &f_{\sfn,2}\left(z^2+\phi z+\psi,p(z),q(z)\right)=\left\lVert \begin{bmatrix}
        P_\calT\\ Q_\calT
    \end{bmatrix}
    \col(1,\phi,\psi)
    \right\rVert_2^2-1\notag\\
    &\,\,=\theta_0\phi^2+\theta_0\psi^2+2\theta_1\phi\psi+2\theta_1\phi+2\theta_2\psi+\theta_0-1.\label{eq:tau2pf}
\end{align}
Then, $[\phi^\star,\,\psi^\star]^\top$ is unique and $\psi^\star=(\theta_1^2-\theta_0\theta_2)/(\theta_0^2-\theta_1^2)$.
Since $\theta_0^2-\theta_1^2>0$, it suffices to show that $\theta_1^2-\theta_0\theta_2< \theta_0^2-\theta_1^2$.
By substituting $\phi=-2\theta_1/\theta_0$ and $\psi=1$ in \eqref{eq:tau2pf}, we have
\begin{equation*}
    \theta_0^{-1}
    \left(2\theta_0^2-4\theta_1^2+2\theta_0\theta_2-\theta_0\right)\geq 0,
\end{equation*}
since the cost function is nonnegative by definition.
As $\theta_0\geq 1$ by definition, it follows that
$\theta_0^2-2\theta_1^2+\theta_0\theta_2\geq \theta_0/2\geq 1/2>0.$
This concludes the proof.

\bibliographystyle{IEEEtran}
\bibliography{ref}

\end{document}

%% file: figures/forward.tex
% This file was created by matlab2tikz.
%
%The latest updates can be retrieved from
%  http://www.mathworks.com/matlabcentral/fileexchange/22022-matlab2tikz-matlab2tikz
%where you can also make suggestions and rate matlab2tikz.
%
\begin{tikzpicture}

\begin{axis}[%
width=0.7\linewidth,
height=0.3\linewidth,
at={(0.758in,0.603in)},
scale only axis,
xmin=0,
xmax=5,
xlabel style={font=\color{white!15!black}},
xlabel={\small Time (s)},
ylabel style={font=\color{white!15!black}},
ylabel={$y$},
ymin=-1,
ymax=2,
axis background/.style={fill=white},
xmajorgrids = true,
ymajorgrids = true,
legend style={at={(0.03,1)}, anchor=north west, legend cell align=left, align=left, draw=white!15!black, legend columns=2},
]
\addplot [color=black, line width=1.8pt]
  table[row sep=crcr]{%
0	0\\
0.05	0.000111704063845713\\
0.1	0.000455755158254372\\
0.15	0.00105084487399185\\
0.2	0.00191432508486234\\
0.25	0.00306214904993991\\
0.3	0.00450881964462603\\
0.35	0.00626734482905435\\
0.4	0.00834920042747024\\
0.45	0.0107643002576547\\
0.5	0.0135209736153914\\
0.55	0.0166259500855416\\
0.6	0.0200843516186375\\
0.65	0.0238996917801627\\
0.7	0.0280738820489926\\
0.75	0.032607245011935\\
0.8	0.0374985342730572\\
0.85	0.0427449608696149\\
0.9	0.0483422259610082\\
0.95	0.0542845595333677\\
1	0.0605647648401954\\
1.05	0.067174268279025\\
1.1	0.0741031743853818\\
1.15	0.0813403256084614\\
1.2	0.0888733665179538\\
1.25	0.096688812078345\\
1.3	0.104772119615853\\
1.35	0.113107764093915\\
1.4	0.121679316305823\\
1.45	0.130469523587755\\
1.5	0.139460392651947\\
1.55	0.148633274138203\\
1.6	0.157968948482241\\
1.65	0.167447712701424\\
1.7	0.177049467702378\\
1.75	0.186753805720526\\
1.8	0.196540097508866\\
1.85	0.206387578902128\\
1.9	0.21627543639282\\
1.95	0.226182891367449\\
2	0.236089282664332\\
2.05	0.245974147128826\\
2.1	0.255817297857356\\
2.15	0.265598899838257\\
2.2	0.275299542715064\\
2.25	0.28490031041634\\
2.3	0.294382847415403\\
2.35	0.303729421403184\\
2.4	0.312922982177931\\
2.45	0.32194721657637\\
2.5	0.330786599292189\\
2.55	0.339426439449196\\
2.6	0.347852922818106\\
2.65	0.356053149587607\\
2.7	0.364015167621869\\
2.75	0.371728001158142\\
2.8	0.379181674919174\\
2.85	0.386367233636029\\
2.9	0.393276756997193\\
2.95	0.399903370059706\\
3	0.406241249177261\\
3.05	0.412285623518753\\
3.1	0.418032772268536\\
3.15	0.423480017616621\\
3.2	0.428625713663114\\
3.25	0.433469231376376\\
3.3	0.438010939758534\\
3.35	0.442252183385124\\
3.4	0.44619525649776\\
3.45	0.449843373839674\\
3.5	0.453200638433874\\
3.55	0.45627200651239\\
3.6	0.459063249812673\\
3.65	0.461580915463595\\
3.7	0.463832283688798\\
3.75	0.465825323559185\\
3.8	0.467568647029331\\
3.85	0.469071461494345\\
3.9	0.470343521104417\\
3.95	0.471395077073859\\
4	0.472236827219945\\
4.05	0.472879864964342\\
4.1	0.473335628026361\\
4.15	0.47361584703282\\
4.2	0.473732494263832\\
4.25	0.473697732747624\\
4.3	0.473523865910336\\
4.35	0.47322328797894\\
4.4	0.472808435326806\\
4.45	0.472291738942248\\
4.5	0.47168557819056\\
4.55	0.471002236029686\\
4.6	0.470253855828891\\
4.65	0.469452399928538\\
4.7	0.468609610067525\\
4.75	0.467736969793086\\
4.8	0.466845668955604\\
4.85	0.465946570378839\\
4.9	0.465050178783681\\
4.95	0.464166612031199\\
};
\addlegendentry{\small true output}

\addplot [color=blue, dashed, line width=2.6pt]
  table[row sep=crcr]{%
0	0\\
0.05	0.000111704063845713\\
0.1	0.000455755158254372\\
0.15	0.00105084487399185\\
0.2	0.00191432508486234\\
0.25	0.00306214904993991\\
0.3	0.00450881964462632\\
0.35	0.00626734482905558\\
0.4	0.00834920042747386\\
0.45	0.0107643002576633\\
0.5	0.0135209736154094\\
0.55	0.0166259500855757\\
0.6	0.0200843516186978\\
0.65	0.0238996917802636\\
0.7	0.0280738820491542\\
0.75	0.0326072450121844\\
0.8	0.0374985342734303\\
0.85	0.042744960870158\\
0.9	0.0483422259617805\\
0.95	0.0542845595344437\\
1	0.0605647648416671\\
1.05	0.0671742682810051\\
1.1	0.0741031743880072\\
1.15	0.0813403256118961\\
1.2	0.0888733665223928\\
1.25	0.0966888120840183\\
1.3	0.10477211962303\\
1.35	0.113107764102908\\
1.4	0.121679316316993\\
1.45	0.130469523601518\\
1.5	0.139460392668773\\
1.55	0.148633274158631\\
1.6	0.157968948506874\\
1.65	0.167447712730944\\
1.7	0.177049467737545\\
1.75	0.186753805762189\\
1.8	0.196540097557964\\
1.85	0.206387578959699\\
1.9	0.216275436460008\\
1.95	0.226182891445507\\
2	0.236089282754632\\
2.05	0.245974147232861\\
2.1	0.255817297976748\\
2.15	0.265598899974763\\
2.2	0.275299542870581\\
2.25	0.284900310592911\\
2.3	0.294382847615223\\
2.35	0.303729421628604\\
2.4	0.312922982431462\\
2.45	0.321947216860689\\
2.5	0.330786599610144\\
2.55	0.339426439803807\\
2.6	0.347852923212573\\
2.65	0.356053150025308\\
2.7	0.364015168106366\\
2.75	0.371728001693182\\
2.8	0.379181675508691\\
2.85	0.386367234284144\\
2.9	0.393276757708219\\
2.95	0.399903370838142\\
3	0.406241250027797\\
3.05	0.412285624446266\\
3.1	0.418032773278092\\
3.15	0.423480018713468\\
3.2	0.428625714852684\\
3.25	0.433469232664281\\
3.3	0.43801094115056\\
3.35	0.442252184887231\\
3.4	0.446195258116074\\
3.45	0.449843375580484\\
3.5	0.453200640303622\\
3.55	0.456272008517673\\
3.6	0.459063251960229\\
3.65	0.4615809177603\\
3.7	0.463832286141654\\
3.75	0.465825326175318\\
3.8	0.467568649815978\\
3.85	0.469071464458847\\
3.9	0.470343524254209\\
3.95	0.471395080416463\\
4	0.472236830762959\\
4.05	0.472879868715428\\
4.1	0.47333563199324\\
4.15	0.473615851223257\\
4.2	0.473732498685629\\
4.25	0.473697737408607\\
4.3	0.473523870818348\\
4.35	0.473223293141826\\
4.4	0.472808440752408\\
4.45	0.472291744638391\\
4.5	0.47168558416504\\
4.55	0.471002242290265\\
4.6	0.470253862383286\\
4.65	0.469452406784408\\
4.7	0.468609617232462\\
4.75	0.467736977274611\\
4.8	0.466845676761154\\
4.85	0.465946578515757\\
4.9	0.465050187259213\\
4.95	0.464166620852482\\
};
\addlegendentry{\small $h^\top=Y_\ff\rmH^\dagger$}

\addplot [color=red, dotted, line width=1.8pt]
  table[row sep=crcr]{%
0	0\\
0.05	0.000111704063845713\\
0.1	0.000455755158254372\\
0.15	0.00105084487399185\\
0.2	0.00191432508486234\\
0.25	0.00306214904993991\\
0.3	0.00450881964462974\\
0.35	0.00626734482906361\\
0.4	0.00834920042749747\\
0.45	0.0107643002577073\\
0.5	0.0135209736154937\\
0.55	0.0166259500857164\\
0.6	0.0200843516189161\\
0.65	0.0238996917806066\\
0.7	0.0280738820496206\\
0.75	0.0326072450128834\\
0.8	0.0374985342743116\\
0.85	0.0427449608713992\\
0.9	0.0483422259633214\\
0.95	0.0542845595364021\\
1	0.0605647648442226\\
1.05	0.0671742682837939\\
1.1	0.0741031743920236\\
1.15	0.0813403256156166\\
1.2	0.0888733665281688\\
1.25	0.0966888120891254\\
1.3	0.104772119630033\\
1.35	0.113107764111027\\
1.4	0.12167931632285\\
1.45	0.13046952361609\\
1.5	0.13946039266938\\
1.55	0.148633274182796\\
1.6	0.157968948501504\\
1.65	0.167447712757777\\
1.7	0.177049467743787\\
1.75	0.186753805754836\\
1.8	0.196540097632536\\
1.85	0.206387578837653\\
1.9	0.216275436689802\\
1.95	0.226182891138646\\
2	0.236089283128037\\
2.05	0.245974146913428\\
2.1	0.255817298032305\\
2.15	0.265598900505956\\
2.2	0.275299541244033\\
2.25	0.284900313778012\\
2.3	0.294382842355548\\
2.35	0.303729428797595\\
2.4	0.312922974167188\\
2.45	0.321947223328308\\
2.5	0.330786599721736\\
2.55	0.339426424738855\\
2.6	0.347852963519953\\
2.65	0.3560530722868\\
2.7	0.364015290909526\\
2.75	0.371727836224246\\
2.8	0.379181856049151\\
2.85	0.386367103880888\\
2.9	0.393276714803705\\
2.95	0.399903773484485\\
3	0.406240241827781\\
3.05	0.412287490803753\\
3.1	0.418029890038701\\
3.15	0.42348378522952\\
3.2	0.428621767510262\\
3.25	0.433471721246869\\
3.3	0.4380128300221\\
3.35	0.442241535307602\\
3.4	0.446220200867121\\
3.45	0.449798652037946\\
3.5	0.453267933878758\\
3.55	0.456186583703199\\
3.6	0.459148503124815\\
3.65	0.461536065151764\\
3.7	0.463767081113434\\
3.75	0.466102345742686\\
3.8	0.466955087173495\\
3.85	0.470138299836999\\
3.9	0.468779063439527\\
3.95	0.473321212823175\\
4	0.470418281286012\\
4.05	0.47360763006857\\
4.1	0.475352340733472\\
4.15	0.466500366244052\\
4.2	0.488739027087867\\
4.25	0.448354007065353\\
4.3	0.509740023241504\\
4.35	0.430060147566456\\
4.4	0.511009984144947\\
4.45	0.462954854197185\\
4.5	0.413055207761075\\
4.55	0.651826121366218\\
4.6	0.105154034869886\\
4.65	1.06906939153423\\
4.7	-0.366094101486246\\
4.75	1.42841171033403\\
4.8	-0.320295829350787\\
4.85	0.501272045646383\\
4.9	2.10242624451871\\
4.95	-4.08914483317077\\
};
\addlegendentry{\small $h^\top=Y_\ff\rmH^\rmG$}

\end{axis}
\end{tikzpicture}%